\documentclass[11pt,epsf,reqno]{amsart} 

\usepackage{times}

\usepackage{amssymb}
\usepackage{amsthm}
\usepackage{amsxtra}
\usepackage{amsmath}

\usepackage{verbatim}
\usepackage{epsfig}
\usepackage{graphicx}

\usepackage{tikz}
\usetikzlibrary{arrows}
\usepackage{pgfplots}

\newtheorem{theorem}{Theorem}[section]

\newtheorem{lemma}[theorem]{Lemma}
\newtheorem{proposition}[theorem]{Proposition}
\newtheorem{corollary}[theorem]{Corollary}

\newtheorem{conjecture}{Conjecture}[section]

\def\Z{{\mathbb Z}}
\def\R{{\mathbb R}}

\def\om{\omega}

\def\S{\Sigma}

\def\w{{\mathbf w}}
\def\x{{\mathbf x}}
\def\y{{\mathbf y}}
\def\z{{\mathbf z}}
\def\0{{\mathbf 0}}
\def\1{{\mathbf 1}}
\def\cC{{\mathcal C}}
\def\cE{{\mathcal E}}

\def\cH{{\mathcal H}}

\def\cX{{\mathcal X}}

\def\sfh{{\mathsf h}}
\def\sfw{{\mathsf w}}

\def\oR{\overline{R}}

\def\tom{\tilde{\omega}}
\def\tm{\tilde{m}}

\def\supp{{\text{supp}}}

\def\beq{\begin{equation}}
\def\eeq{\end{equation}}
\def\beqa{\begin{eqnarray*}}
\def\eeqa{\end{eqnarray*}}

\def\disj{\stackrel{\cdot}{\cup}}

\pgfplotsset{
    standard/.style={
        axis lines=middle,
        every axis x label/.style={at={(current axis.right of origin)},anchor=north west},
        every axis y label/.style={at={(current axis.above origin)},anchor=north east}
    }
}

\title[Upper Bounds on the Size of Grain-Correcting Codes]{Upper Bounds on the Size of Grain-Correcting Codes$^*$}

\author[Navin Kashyap and Gilles Z{\'e}mor]{Navin Kashyap$^\dag$ \and Gilles Z{\'e}mor$^\ddag$}

\begin{document}
\maketitle

\renewcommand{\thefootnote}{}
\footnotetext{
$^*$Parts of this paper were submitted to the 2013 IEEE International Conference on Information Theory (ISIT 2013),
to be held in Istanbul, Turkey,  July 7--12, 2013.

$^\dag$Department of Electrical Communication Engineering, 
Indian Institute of Science, Bangalore 560012. Email: nkashyap@ece.iisc.ernet.in

\smallskip

$^\ddag$Institut de Math{\'e}matiques de Bordeaux, UMR 5251, Universit{\'e} Bordeaux 1 --- 
351, cours de la Lib{\'e}ration --- 33405 Talence Cedex, France. Email: gilles.zemor@math.u-bordeaux1.fr

}

\renewcommand{\thefootnote}{\arabic{footnote}}
\setcounter{footnote}{0}

\begin{abstract}
In this paper, we re-visit the combinatorial error model of Mazumdar et al.\ \cite{MBK11} that models errors 
in high-density magnetic recording caused by lack of knowledge of grain boundaries in the recording medium.
We present new upper bounds on the cardinality/rate of binary block codes that correct errors within this model.
\end{abstract}

\section{Introduction\label{sec:intro}} 
The combinatorial error model studied by Mazumdar et al.\ \cite{MBK11} is a highly simplified model 
of an error mechanism encountered in a magnetic recording medium 
at terabit-per-square-inch storage densites \cite{Pan_etal}, \cite{Wood_etal}. In this model, 
a one-dimensional track on a magnetic recording medium is divided into evenly spaced bit cells, 
each of which can store one bit of data. Bits are written sequentially into these bit cells. 
The sequence of bit cells has an underlying ``grain'' distribution, which may be described as follows: 
bit cells are grouped into non-overlapping blocks called \emph{grains}, which may consist of up to $b$ 
adjacent bit cells. We focus on the case $b=2$, so that a grain can contain at most two bit cells.
We define the \emph{length} of a grain to be the number of bit cells it contains. 

Each grain can store only one bit of information, i.e., all the bit cells within a grain carry the same bit value (0 or 1), 
which we call the \emph{polarity} of the grain. 
We assume, following \cite{MBK11}, that in the sequential write process, the first bit to be written into a grain 
sets the polarity of the grain, so that all the bit cells within this grain must retain this polarity.\footnote{Considering 
the physics of the write process, it would make more sense to assume that the \emph{last} bit to be 
written within a grain sets the polarity of the grain, 
thus overwriting all other bits previously stored in the bit cells comprising the grain. However, mathematically, 
this is completely equivalent to the polarity-set-by-first-bit model.} This implies 
that any subsequent attempts at writing bits within this grain make no difference to the value actually stored in 
the bit cells in the grain. If the grain boundaries were known to the write head (encoder) and the read head (decoder),
then the maximum storage capacity of one bit per grain can be achieved. However, in a more realistic scenario 
where the underlying grain distribution is fixed but \emph{unknown}, the lack of knowledge of grain boundaries
reduces the storage capacity. Constructions and rate/cardinality bounds for codes that correct errors caused by 
a fixed but unknown underlying grain distribution have been studied in the prior literature \cite{MBK11}, \cite{SR11}, \cite{SR13}. In this paper, we present improved rate/cardinality upper bounds for such codes.

The paper is organized as follows. After providing the necessary definitions and notation in Section~\ref{sec:defs},
we derive, in Section~\ref{sec:Mnt_bnd}, an upper bound on the cardinality of $t$-grain-correcting codes, for $t = 1,2,3$, using the fractional covering technique from \cite{KK12}. We conjecture that the upper bound in fact holds
for all $t$. We further conjecture that the same technique should yield a stronger upper bound, and we report some 
progress towards this in Section~\ref{sec:Mnt_bnd2}. The fractional covering technique is also used in Section~\ref{sec:Rtau_upbnd1} to obtain an upper bound on the maximum rate asymptotically achievable by codes correcting a constant fraction of grain errors. An information-theoretic upper bound on the same quantity is derived 
in Section~\ref{sec:Rtau_upbnd2}. We conclude in Section~\ref{sec:conclusion} with some remarks concerning 
the bounds. The current state-of-the-art on upper bounds on the maximum rate asymptotically achievable, including the bounds derived in this paper, is summarized in Figures~\ref{fig:upper} and \ref{fig:conj4.1}. Some of the technical proofs from Sections~\ref{sec:Mnt_bnd}  and \ref{sec:Rtau_upbnd1} are given in appendices.

\section{Definitions and Notation\label{sec:defs}}

Let $\S = \{0,1\}$, and for a positive integer $n$, let $[n]$ denote the set $\{1,2,\ldots,n\}$. A track on the
recording medium consists of $n$ bit cells indexed by the integers in $[n]$. The bit cells on the track are grouped
into non-overlapping grains of length at most two. A length-2 grain consists of bit cells with indices $j-1$ and $j$,
for some $j \in [n]$; we denote such a grain by the pair $(j-1,j)$. Let $E \subseteq \{2,\ldots,n\}$ be the set of all 
indices $j$ such that $(j-1,j)$ is a length-2 grain. Since grains cannot overlap, $E$ contains no pair of consecutive 
integers.

A binary sequence $\x = (x_1,\ldots,x_n) \in \S^n$ to be written on to the track can be affected by errors only at 
the indices $j \in E$. Indeed, what actually gets recorded on the track is the sequence $\y = (y_1,\ldots,y_n)$,
where 
\begin{equation}
y_j = 
\begin{cases}
x_{j-1}  & \text{ if } j \in E \\
x_j & \text{ otherwise. } 
\end{cases}
\label{eq:xy}
\end{equation}
For example, if $\x = (000101011100010)$ and $E = \{2,4,7,9,14\}$, then $\y = (000001111100000)$. 
Note that the set $E$ completely specifies the positions and locations of all the grains 
(both length-1 and length-2) in the track. We will call this set the 
\emph{grain pattern}. It is assumed that the grain pattern is unknown to both the write head and the read head.
The effect of the grain pattern $E$ on a binary sequence $\x \in \S^n$ defines an operator 
$\phi_E: \S^n \to \S^n$, where $\y = \phi_E(\x)$ is as specified by (\ref{eq:xy}) above.

For integers $n \ge 1$ and $t \ge 0$, let $\cE_{n,t}$ denote the set of all grain patterns $E$ with $|E| \le t$. In other
words, $\cE_{n,t}$ consists of all subsets $E \subseteq \{2,\ldots,n\}$ of cardinality at most $t$, 
such that $E$ contains no pair of consecutive integers. For an $\x \in \S^n$, we define
$$
\Phi_t(\x) = \{\phi_E(\x): E \in \cE_{n,t}\}.
$$
Thus, $\Phi_t(\x)$ is the set of all possible sequences that can be obtained from $\x$ by the action of some 
grain pattern $E$ with $|E| \le t$.
Two sequences $\x_1$ and $\x_2$ are \emph{$t$-confusable} if $\Phi_t(\x_1) \cap \Phi_t(\x_2) \neq \emptyset$. 
A binary code $\cC$ of length $n$ is said to correct $t$ grain errors, or be a \emph{$t$-grain-correcting code},
if no two distinct vectors $\x_1,\x_2 \in \cC$ are $t$-confusable. Let $M(n,t)$ denote the maximum cardinality of 
a $t$-grain-correcting code of length $n$. Also, for $\tau \in [0,\frac12]$, the maximum asymptotic rate of 
a $\lceil{\tau n}\rceil$-grain-correcting code is defined to be 
\beq
R(\tau) = \limsup_{n\to\infty} \frac{1}{n} \, \log_2 M(n,\lceil \tau n \rceil).
\label{eq:Rtau}
\eeq

A grain pattern $E$ changes a sequence $\x$ to a different sequence $\y$ iff $x_{j-1} \ne x_j$ for some $j \in E$,
i.e., the length-2 grain $(j-1,j)$ straddles the boundary between two successive runs in $\x$. 
Here, a \emph{run} is a maximal 
substring of consecutive identical symbols ($0$s or $1$s) in $\x$. A run consisting of $0$s (resp.\ $1$s) is called
a \emph{$0$-run} (resp.\ \emph{$1$-run}). The number of distinct runs in $\x$ is denoted by $r(\x)$. 

A convenient means of keeping track of run boundaries in $\x$ is via its \emph{derivative sequence}, $\x'$: 
for $\x = (x_1,\ldots,x_n)$, the sequence $\x' = (x'_2,\ldots,x'_{n})$
is defined by $x'_j = x_{j-1} \oplus x_j$, $j = 2,\ldots,n$, where $\oplus$ denotes modulo-2 addition. 
The $1$s in $\x'$ identify the boundaries between successive runs in $\x$. Thus, $\omega(\x') = r(\x) - 1$,
where $\omega(\cdot)$ denotes the Hamming weight of a binary sequence. 

Let $\supp(\x') = \{j: x'_j = 1\}$ denote the support of $\x'$. 
For $\x \in \S^n$, the sequences $\y \in \Phi_t(\x)$ are in one-to-one correspondence with the 
different ways of selecting at most $t$ non-consecutive integers\footnote{A sequence or set of 
non-consecutive integers is one that does not contain a pair of consecutive integers.} 
from $\supp(\x')$ to form a grain pattern $E \in \cE_{n,t}$. 
Thus, $|\Phi_t(\x)|$ counts the number of ways of forming such grain patterns. This count can be obtained as 
follows. Let $\ell_1, \ell_2, \ldots, \ell_m$ denote the lengths of the distinct $1$-runs in $\x'$, and define the set 
\beq
T = \bigl\{(t_1,\ldots,t_m) \in \Z_+^m: \sum_{j=1}^m t_j \le t\bigr\},
\label{eq:T}
\eeq
where $\Z_+$ denotes the set of non-negative integers. 
In the above expression, $t_j$ represents the number of integers from the support of the $j$th 1-run that are to be 
included in a grain pattern $E$ being formed.  The number of distinct ways in which $t_j$ non-consecutive integers can 
be chosen from the $\ell_j$ consecutive integers forming the support of the $j$th $1$-run is, by an elementary 
counting argument, equal to $\binom{\ell_j - t_j + 1}{t_j}$. Thus,
\begin{equation}
|\Phi_t(\x)| = \sum_{(t_1,\ldots,t_m) \in T} \prod_{j=1}^m \binom{\ell_j - t_j + 1}{t_j}.
\label{eq:Phi_t}
\end{equation}
Simplified expressions can be obtained for small values of $t$. 

\begin{proposition}
For $\x \in \S^n$, let $\omega = \omega(\x')$ denote the Hamming weight of the derivative sequence $\x'$. 
Also, let $m$ be the number of 1-runs in $\x'$. 
\begin{itemize}
\item[(a)] $|\Phi_1(\x)| = 1 + \omega = r(\x)$. 
\item[(b)] $|\Phi_2(\x)| = 1 + m + \binom{\omega}{2}$. 
\item[(c)] $|\Phi_3(\x)| = 1 + m_1 + m(\omega-3) + \binom{\omega}{3} - \binom{\omega}{2} + 2\omega$,
where $m_1$ denotes the number of 1-runs of length 1 in $\x'$.
\end{itemize}
\label{prop:small_t}
\end{proposition}

\begin{proof}
(a)\ While the expression for $|\Phi_1(\x)|$ can be directly obtained from (\ref{eq:Phi_t}), it is simpler to observe that 
the set $\Phi_1(\x)$ consists of the sequence $\x$ itself, and the $\omega$ distinct sequences in the set 
$\bigl\{\phi_E(\x): E = \{j\} \text{ for some } j \in \supp(\x')\bigr\}$. 

(b)\ For $t=2$, it is easy to see that the expression in (\ref{eq:Phi_t}) simplifies to
$$
|\Phi_2(\x)| = 1 + \sum_{j=1}^m \ell_j + \sum_{j=1}^m \binom{\ell_j-1}{2} + \sum_{(i,j):i < j} \ell_i \ell_j.
$$
We then have
\beqa
|\Phi_2(\x)| 
&=& 1 + m + \sum_{j=1}^m (\ell_j-1) + \sum_{j=1}^m \binom{\ell_j-1}{2} + \sum_{(i,j):i < j} \ell_i \ell_j \\
&=& 1 + m + \sum_{j=1}^m \binom{\ell_j}{2} + \sum_{(i,j):i < j} \ell_i \ell_j \\
&=& 1 + m + \frac12 \sum_{j=1}^m (\ell_j^2 - \ell_j) + \sum_{(i,j):i < j} \ell_i \ell_j \\
&=& 1 + m + \frac12 \left[\sum_{j=1}^m \ell_j^2 + \sum_{(i,j): i \ne j} \ell_i \ell_j - \sum_{j=1}^m \ell_j\right] \\
&=& 1 + m + \frac12 \left[\biggl(\sum_{j=1}^m \ell_j\biggr)^2 - \sum_{j=1}^m \ell_j\right] \\
&=& 1 + m + \binom{\omega}{2}, 
\eeqa
the last equality being due to the fact that $\omega = \sum_{j=1}^m \ell_j$.

(c)\ For $t = 3$, the expression in (\ref{eq:Phi_t}) can be written as
$$
|\Phi_3(\x)| = |\Phi_2(\x)| + \sum_{(i,j,k): i < j< k} \ell_i \ell_j \ell_k 
+ \sum_{i=1}^m \binom{\ell_i - 1}{2}\biggl(\sum_{j: j \ne i} \ell_j\biggr) + \sum_{i=1}^m \binom{\ell_i - 2}{3}.
$$
From here on, straightforward algebraic manipulations lead to the expression given in the statement of the 
proposition. We omit the details, noting only that $m_1$ enters the picture when we write the last term above as
$$
\sum_{i=1}^m \frac{(\ell_i-2)(\ell_i-3)(\ell_i-4)}{6} + \sum_{i: \ell_i = 1} 1.
$$
The extra term $\sum_{i: \ell_i = 1} 1$, which equals $m_1$, accounts for the fact that the expansion of 
$\binom{\ell_i -2}{3}$ as $\frac{(\ell_i-2)(\ell_i-3)(\ell_i-4)}{6}$ is invalid when $\ell_i = 1$; by convention,
$\binom{a}{b} = 0$ when $a < 0$.
\end{proof}

\medskip

We will also find the following simple lower bound on $|\Phi_t(\x)|$, valid for any $t \ge 1$, to be useful. 

\begin{proposition}
For $\x \in \S^n$ and $t \ge 1$, we have
$$
|\Phi_t(\x)| \ge \sum_{j=0}^t \binom{r(\x)-j}{j}.
$$
\label{prop:Phi_lobnd}
\end{proposition}
\begin{proof} Consider the number of different ways of choosing exactly $j$ non-consecutive integers from $\supp(\x')$. This number is smallest when $\supp(\x')$ consists of consecutive integers, e.g., $\supp(\x) = [r(\x)-1]$. The number of different ways of choosing exactly $j$ non-consecutive integers from $[r(\x)-1]$ is, by an elementary counting argument, equal to $\binom{r(\x) - j}{j}$. 
\end{proof}

\section{An Upper Bound on $M(n,t)$\label{sec:Mnt_bnd}}

In this section, we explore the applicability to grain-correcting codes of a technique used by Kulkarni and
Kiyavash \cite{KK12} to derive upper bounds on the cardinalities of deletion-correcting codes. 

A \emph{hypergraph} $\cH$ is a pair $(V,\cX)$, where $V$ is a finite set, called the \emph{vertex set}, and $\cX$
is a family of subsets of $V$. The members of $\cX$ are called \emph{hyperedges}. A \emph{matching}
of $\cH$ is a pairwise disjoint collection of hyperedges. A \emph{(vertex) covering} of $\cH$ is a subset 
$T \subseteq V$ such that $T$ meets every hyperedge of $\cH$, i.e., $T \cap X \ne \emptyset$ for all $X \in \cX$.
The \emph{matching number} $\nu(\cH)$ is the largest size of a matching of $\cH$, while the 
\emph{covering number}, $\tau(\cH)$, is the smallest size of a covering of $\cH$. 

The problems of computing the matching and covering numbers can be expressed as a 
dual pair of integer programs. This is done via the vertex-hyperedge incidence matrix, $A$, of $\cH$, which
is defined as follows. Let $v_1,v_2,\ldots,v_{|V|}$ and $X_1,X_2,\ldots,X_{|\cX|}$ be a listing of the vertices and
hyperedges, respectively, of $\cH$. Then, $A = (A_{i,j})$ is the $|V| \times |\cX|$ matrix with $0/1$ entries, with
$A_{i,j} = 1$ iff $v_i \in X_j$. It is easy to verify that 
\beq
\nu(\cH) = \max\{\1^T\z: \z \in \{0,1\}^{|\cX|},\ A \z \le \1\}  
\label{eq:nu}
\eeq
and
\beq
\tau(\cH) = \min\{\1^T\w: \w \in \{0,1\}^{|V|},\ A^T \w \ge \1\},  
\label{eq:tau}
\eeq
where $\1$ denotes an all-ones column vector. Note that the linear programming (LP) relaxations of (\ref{eq:nu}),
\beq
\nu_f(\cH) = \max\{\1^T\z: \z \ge \0,\ A \z \le \1\}, 
\label{eq:nu_f}
\eeq
and (\ref{eq:tau}),
\beq
\tau_f(\cH) = \min\{\1^T\w: \w \ge \0,\ A^T \w \ge \1\},  
\label{eq:tau_f}
\eeq
are duals of each other. By strong LP duality, we have $\nu_f(\cH) = \tau_f(\cH)$, and hence,
\beq
\nu(\cH) \le \nu_f(\cH) = \tau_f(\cH) \le \tau(\cH).
\label{duality}
\eeq

The quantities $\nu_f(\cH)$ and $\tau_f(\cH)$ are called the \emph{fractional matching number} 
and \emph{fractional covering number}, respectively, of the hypergraph $\cH$. Any non-negative vector
$\w$ such that $A^T \w \ge 1$ is called a \emph{fractional covering}\footnote{A fractional matching
is correspondingly defined, but we will have no further use for this concept.} of $\cH$. To put it in another
way, a fractional covering is a function $\sfw: V \to \R_+$ such that $\sum_{v \in X} \sfw(v) \ge 1$ for all $X \in \cX$.
The \emph{value} of a fractional covering $\sfw$ is defined to be $|\sfw| := \sum_{v \in V} \sfw(v)$. From the 
inequality $\nu(\cH) \le \tau_f(\cH)$ in (\ref{duality}), we see that $\nu(\cH) \le |\sfw|$ for any fractional covering
$\sfw$ of $\cH$. We use this to suggest an upper bound on the largest size, $M(n,t)$, of a $t$-grain-correcting
code of blocklength $n$. 

Consider the hypergraph $\cH_{n,t} = (V,\cX)$, where $V = \S^n$, and $\cX = \{\Phi_t(\x): \x \in \S^n\}$. Note that
$\nu(\cH_{n,t}) = M(n,t)$; thus, fractional coverings of $\cH_{n,t}$ yield upper bounds on $M(n,t)$. 
Bounding the size of packings in this way has been extensively used in combinatorics, see e.g. \cite{Ber79}. 
Taking inspiration from \cite{KK12}, we consider the function 
$\sfw_t: \S^n \to \R_+$, defined for $\x \in \S^n$ as 
\beq
\sfw_t(\x) = \frac{1}{|\Phi_t(\x)|}.
\label{eq:w_t}
\eeq
For $t = 1,2,3$, we can prove that $\sfw_t$ is a fractional covering of $\cH_{n,t}$, and conjecture that
this is in fact the case for all $t \ge 1$. 

\begin{conjecture}
For all positive integers $n$ and $t$, the function $\sfw_t$ defined in (\ref{eq:w_t}) is a fractional covering of 
$\cH_{n,t}$, i.e., for all $\x \in \S^n$,
\beq
\sum_{\y \in \Phi_t(\x)} \frac{1}{|\Phi_t(\y)|} \ge 1.
\label{eq:sum_y}
\eeq
Therefore, 
\beq
M(n,t) \le |\sfw_t| = \sum_{\x\in\S^n} \frac{1}{|\Phi_t(\x)|}.
\label{Mnt_upbnd1}
\eeq
\label{conj:Phi_t}
\end{conjecture}

Our proof of (\ref{eq:sum_y}) for $t =1,2,3$ relies on an understanding of the relationship between $|\Phi_t(\x)|$
and $|\Phi_t(\y)|$ for $\y \in \Phi_t(\x)$. Recall, from (\ref{eq:Phi_t}), that $|\Phi_t(\x)|$ depends only on the 
lengths of the 1-runs in $\x'$. Thus, we need to understand how the distribution of 1s changes in going from
$\x'$ to $\y'$. 

\subsection{Effect of Grains on the Derivative Sequence\label{sec:derivative}}

Recall that $1$s in $\x'$ correspond to run boundaries in $\x$. We say that a (length-2) grain \emph{acts on} 
a $1$ in $\x'$ if it straddles the corresponding run boundary in $\x$. 
We need to distinguish between two types of 1s in the derivative sequence $\x'$. A \emph{trailing $1$} is the last
$1$ in a $1$-run, while a \emph{non-trailing $1$} is any $1$ that is not a trailing $1$. Grains act on trailing $1$s
in a manner different from non-trailing $1$s. 

A segment of $\x'$ that contains a trailing $1$ is of the form $*10*$, or $*1$ in case the trailing $1$ is 
a suffix of $\x$. Up to complementation, the corresponding segment of $\x$ is of the form $*011*$ or $*01$. 
A grain acting on the trailing $1$ in $\x'$ straddles the $01$ run boundary in $\x$. In the sequence $\y$ obtained
through the action of this grain, the segment under observation becomes $*001*$ or $*00$, and the
corresponding segment of the derivative sequence $\y'$ is $*01*$ or $*0$. 

On the other hand, a non-trailing $1$ in $\x'$ belongs to a segment of the form $*11*$; the first $1$ shown
is the non-trailing $1$ under consideration. Again, up to complementation, the corresponding segment in $\x$ 
is of the form $*010*$. A grain acting on the non-trailing $1$ in $\x'$ straddles the $01$ run boundary shown in 
$\x$. This grain causes the segment being observed to become $*000*$ in $\y$, and hence $*00*$ in $\y'$.

To summarize, the action of a grain on a trailing $1$ converts a segment of the form $*10*$ or $*1$ in $\x'$
to $*01*$ or $*0$ in $\y'$, and a grain acting on a non-trailing $1$ converts a segment of the form $*11*$ in $\x'$
to $*00*$ in $\y'$. It should be clear that the bits depicted by $*$s on either side of these segments 
remain unchanged by the action of the grain. Note, in particular, that a grain acting on a $1$ in $\x'$ does not 
increase the Hamming weight of $\x'$. A grain acting on a trailing $1$ either leaves the Hamming weight of $\x'$ 
unchanged, or reduces it by $1$; in the case of a non-trailing $1$, the Hamming weight of $\x'$ is always
reduced by $2$. 

Finally, when dealing with a grain pattern containing $t > 1$ length-2 grains, since the grains are 
non-overlapping, the actions of individual grains can be considered independently. 
Thus, the discussion above immediately implies the following useful fact.

\begin{lemma}
For any $\y \in \Phi_t(\x)$, we have $\om(\y') \le \om(\x')$, or equivalently, $r(\y) \le r(\x)$. 
\label{lem:runs}
\end{lemma}

\subsection{Proof of (\ref{eq:sum_y}) for $t = 1,2,3$\label{sec:proof1}}

Consider $t = 1$ first. 
For any $\y \in \Phi_1(\x)$, we have $\omega(\y') \le \omega(\x')$, and hence, 
$|\Phi_1(\y)| \le |\Phi_1(\x)|$ by Proposition~\ref{prop:small_t}. Therefore, 
$$
\sum_{\y \in \Phi_1(\x)} \frac{1}{|\Phi_1(\y)|} \ge \sum_{\y \in \Phi_1(\x)} \frac{1}{|\Phi_1(\x)|} = 1,
$$
which proves $(\ref{eq:sum_y})$ for $t = 1$.

The simple argument above does not extend directly to $t \ge 2$, the reason being that it is no longer true in general
that $|\Phi_t(\y)| \le |\Phi_t(\x)|$ for $\y \in \Phi_t(\x)$. For example, consider $\x = 0100$, and note that 
$\Phi_2(\x) = \{0000,0100,0110\}$. Take $\y = 0110 \in \Phi_2(\x)$, and verify that $\Phi_2(\y) = 
\{0110,0010,0111,0011\}$. Thus, $|\Phi_2(\y)| > |\Phi_2(\x)|$.

To prove (\ref{eq:sum_y}) for $t=2,3$, we show that the sequences $\y \in \Phi_t(\x)$ that violate the inequality
$|\Phi_t(\y)| \le |\Phi_t(\x)|$ can be dealt with by suitably matching them with sequences that satisfy the inequality. 
To this end, for a fixed $\x \in \S^n$, let us define $F_t(\x) = \{\y \in \Phi_t(\x): |\Phi_t(\y)| > |\Phi_t(\x)|\}$
and $G_t(\x) = \{\y \in \Phi_t(\x): |\Phi_t(\y)| \le |\Phi_t(\x)|\}$. We will construct a one-to-one mapping 
$p:F_t(\x) \to G_t(\x)$ such that for all $\y \in F_t(\x)$, we have
\beq
\frac{1}{|\Phi_t(\y)|} + \frac{1}{|\Phi_t(p(\y))|} \ge \frac{2}{|\Phi_t(\x)|}.
\label{pairing}
\eeq
The mapping $p$ will be referred to as a \emph{pairing}. Let $P_t(\x) = p(F_t(\x))$ denote the image of $p$, 
and let $Q_t(\x) = G_t(\x) \setminus P_t(\x)$. Thus, $|P_t(\x)| = |F_t(\x)|$, 
and $\Phi_t(\x) = F_t(\x) \disj P_t(\x) \disj Q_t(\x)$. Then,
\beqa
\sum_{\y \in \Phi_t(\x)} \frac{1}{|\Phi_t(\y)|} 
 &=& \sum_{\y \in F_t(\x)} \frac{1}{|\Phi_t(\y)|} + \sum_{\y \in P_t(\x)} \frac{1}{|\Phi_t(\y)|} + \sum_{\y \in Q_t(\x)} \frac{1}{|\Phi_t(\y)|} \\
 &=& \sum_{\y \in F_t(\x)} \frac{1}{|\Phi_t(\y)|} + \sum_{\y \in F_t(\x)} \frac{1}{|\Phi_t(p(\y))|} + \sum_{\y \in Q_t(\x)} \frac{1}{|\Phi_t(\y)|} \\
 &\ge& \sum_{\y \in F_t(\x)} \frac{2}{|\Phi_t(\x)|} + \sum_{\y \in Q_t(\x)} \frac{1}{|\Phi_t(\x)|}  \\
 &=& \frac{1}{|\Phi_t(\x)|} \, (2 |F_t(\x)| + |Q_t(\x)|).
\eeqa
The last expression above is equal to $1$ since $2 |F_t(\x)| + |Q_t(\x)| 
= |F_t(\x)| + |P_t(\x)| + |Q_t(\x)| = |\Phi_t(\x)|$. Thus, the construction of a pairing satisfying (\ref{pairing}) 
is sufficient to prove (\ref{eq:sum_y}), and hence, (\ref{Mnt_upbnd1}). 
Such a pairing can indeed be constructed for $t = 2,3$, and we give a proof of this in Appendix~A.

In summary, we have obtained the following result.

\begin{theorem}
For any positive integer $n$ and $t=1,2,3$, we have
$$
M(n,t) \le \sum_{\x \in \S^n} \frac{1}{|\Phi_t(\x)|}.
$$
\label{thm:Mnt_upbnd1}
\end{theorem}
For $t=1$, an exact closed-form expression can be derived for $\sum_{\x} \frac{1}{|\Phi_t(\x)|}$. Indeed,
\beqa
\sum_{\x \in \S^n} \frac{1}{|\Phi_1(\x)|} 
 &\stackrel{(a)}{=}& \sum_{\x \in \S^n} \frac{1}{|r(\x)|} 
 \ = \ \sum_{r=1}^n \sum_{\x:r(\x) = r} \frac{1}{r} \\
 &\stackrel{(b)}{=}& \sum_{r=1}^n 2\, \binom{n-1}{r-1} \, \frac{1}{r} 
 \ \stackrel{(c)}{=} \ 2\, \sum_{r=1}^n \frac{1}{n} \, \binom{n}{r} \\
&=& \frac{2}{n} (2^n-1) \ = \ \frac{1}{n}(2^{n+1}-2).
\eeqa
Equality~(a) above is by virtue of Proposition~\ref{prop:small_t}; (b) is due to the fact that the number of
$\x \in \S^n$ with $r(\x) = r$ is equal to twice the number of $\x' \in \S^{n-1}$ with $\omega(\x') = r-1$;
and (c) uses the identity $\frac{1}{r} \binom{n-1}{r-1} = \frac{1}{n} \binom{n}{r}$.
Thus, we have
\begin{corollary}
$M(n,1) \le  \frac{1}{n}(2^{n+1}-2)$ for all positive integers $n$. 
\label{cor:Mn1}
\end{corollary}

For $t=2,3$, analogous closed-form expressions for the upper bound in Theorem~\ref{thm:Mnt_upbnd1} 
do not appear to exist. However, using Proposition~\ref{prop:small_t}, the bounds can be expressed in a form 
more convenient for numerical evaluation.

\begin{corollary} With the convention that $\binom{a}{-1}$ equals $1$ if $a = -1$, and equals $0$ otherwise, 
the following bounds hold:
\begin{itemize}
\item[(a)] $\displaystyle M(n,2) \le 2 \cdot \sum_{\om = 0}^{n-1} \sum_{m=0}^{\min\{\om,n-\om\}} 
     \binom{\om-1}{m-1} \binom{n-\om}{m} \frac{1}{1+m+\binom{w}{2}}$ \\
\item[(b)] $\displaystyle M(n,3) \le 2 \cdot \sum_{\om = 0}^{n-1} \sum_{m=0}^{\min\{\om,n-\om\}} \sum_{m_1 = 0}^m
     \binom{m}{m_1} \binom{\om-m-1}{m-m_1-1} \binom{n-\om}{m}\frac{1}{\phi_3(m_1,m,\om)}$, where 
$\phi_3(m_1,m,\om) = 1+m_1+m(\om-3) + \binom{\om}{3} - \binom{\om}{2} + 2\om$.
\end{itemize}
\label{cor:Mn23}
\end{corollary}
\begin{proof}
The expressions for the upper bounds are simply alternative ways of expressing $\sum_{\x} \frac{1}{|\Phi_t(\x)|}$ 
using Proposition~\ref{prop:small_t}. The factor 2 in the bounds arises from the fact that each
$\x' \in \S^{n-1}$ is the derivative of exactly two distinct sequences $\x \in \S^n$. In the bound for
$M(n,2)$, the term $\binom{\om-1}{m-1} \binom{n-\om}{m}$ is the number of sequences $\x' \in \S^{n-1}$
with Hamming weight $\om$ and exactly $m$ $1$-runs. Analogously, in the bound for $M(n,3)$,
the term $\binom{m}{m_1} \binom{\om-m-1}{m-m_1-1} \binom{n-\om}{m}$ is the number of sequences 
$\x' \in \S^{n-1}$ with Hamming weight $\om$ and exactly $m$ $1$-runs, of which exactly $m_1$ runs are of 
length $1$.
\end{proof}

\begin{table}[ht]
\begin{center}
\begin{tabular}{|c||c|c|c|c|c|c|c|c|c|c|c|} \hline
$\overset{\text{\normalsize $\;\;\;n$}}{\underset{\text{\normalsize $t\;\;\;$}}{}}$ & 2 & 3 & 4 & 5 & 6 & 7 & 8 & 9 & 10 & 15 & 20 \\ \hline \hline
1 & 3 (2) & 4 (4) & 7 (6) & 12 (8) & 21 (16) & 36 (26) & 63 (44) & 113 & 204 & 4368 & 104857 \\ \hline
2 &   &   & 7 (4) & 11 (8) & 17 (10) & 27 (16) & 43 (22) & 70 & 114 & 1552 & 26418 \\ \hline
3 & & & & & 17 (8) & 26 (16) & 41 (18) & 65 (32) & 101 & 1024 & 12510 \\ \hline
\end{tabular}
\end{center}
\caption{Some numerical values of the upper bound of Theorem~\ref{thm:Mnt_upbnd1}, rounded down to the nearest integer. Within parentheses are the corresponding lower bounds from Table~I of \cite{SR11}.}
\label{table:upbnd}
\end{table}

 Table~\ref{table:upbnd} lists the numerical values of the bounds in Corollaries~\ref{cor:Mn1} and \ref{cor:Mn23} 
for some small values of $n$. Two other upper bounds on $M(n,t)$ exist in the prior literature, namely Corollary~6 
of \cite{MBK11} and Theorem~3.1 of \cite{SR11}. Numerical computations for $n \le 20$ show that
our bounds above are consistently better than the bounds obtained from \cite[Theorem~3.1]{SR11}.
On the other hand, the bound of \cite[Corollary~6]{MBK11} may be better than our bound for small values of $n$:
for example, the bound in \cite{MBK11} yields $M(10,2) \le 92$. However, our bound is better for all $n$ 
sufficiently large: for $t = 1$, our bound is better for all $n \ge 8$; for $t=2$, our bound wins for $n \ge 13$.

\subsection{Some Remarks on the Proof for Arbitrary $t$} We outline here one possible approach 
to proving Conjecture~\ref{conj:Phi_t} for general $t$. To prove (\ref{eq:sum_y}), it is enough to show that
for each $\x \in \S^n$,
$$
\sum_{\y \in \Phi_t(\x)} [|\Phi_t(\y)| - |\Phi_t(\x)|] \le 0.
$$
Indeed, the above inequality is equivalent to showing that the arithmetic mean
$\frac{1}{|\Phi_t(\x)|} \, \sum_{\y \in \Phi_t(\x)} |\Phi_t(\y)|$ is at most $|\Phi_t(\x)|$. If this is true, 
then by concavity of the function $f(x) = \frac{1}{x}$, we would have
$$
\frac{1}{|\Phi_t(\x)|} \, \sum_{\y \in \Phi_t(\x)} \frac{1}{|\Phi_t(\y)|} 
\ge \frac{1}{\frac{1}{|\Phi_t(\x)|} \, {\displaystyle \sum_{\y \in \Phi_t(\x)}} |\Phi_t(\y)|} \ge \frac{1}{|\Phi_t(\x)|}
$$
which is the desired inequality (\ref{eq:sum_y}). The arguments given in Appendix~A for $t=2,3$ essentially
follow this approach.

\section{A Stronger Upper Bound on $M(n,t)$\label{sec:Mnt_bnd2}}

We in fact conjecture that a bound tighter than that of Conjecture~\ref{conj:Phi_t} may hold. To state this bound, 
let us define $V(n,t)$ to be the cardinality of a Hamming ball of radius $t$ in $\S^n$:
$$
V(n,t) = \sum_{j=0}^t \binom{n}{j}.
$$
Note that for any $\x \in \S^n$, we have $|\Phi_t(\x)| \le V(\omega(\x'),t)$, where $\omega(\x')$ is the 
Hamming weight of the derivative sequence $\x'$. This is because $V(\om(\x'),t)$ counts the number of ways 
that a pattern of up to $t$ length-$2$ grains could affect $\x$ if the grains were not constrained to be 
non-overlapping. 

We conjecture that the function $\widetilde{\sfw}_t: \S^n \to \R_+$, defined by 
\beq
\widetilde{\sfw}_t(\x) = \frac{1}{V(\omega(\x'),t)}
\label{eq:wtilde}
\eeq
is a fractional covering of the hypergraph  $\cH_{n,t}$. Note that 
$$
|\widetilde{\sfw}_t| = \sum_{\x\in\S^n} \frac{1}{V(\om(\x'),t)} = 
2 \, \sum_{\om = 0}^{n-1} \binom{n-1}{\om} \, \frac{1}{V(\om,t)},
$$
since $2\binom{n-1}{\om}$ is the number of sequences $\x \in \S^n$ whose derivative sequence $\x'$ 
has Hamming weight $\om$.

\begin{conjecture}
For all positive integers $n$ and $t$, and for all $\x \in \S^n$, we have
\beq
\sum_{\y \in \Phi_t(\x)} \frac{1}{V(\om(\y'),t)} \ge 1.
\label{eq:sum_y2}
\eeq
Therefore, 
\beq
M(n,t) \le  2 \, \sum_{\om = 0}^{n-1} \binom{n-1}{\om} \, \frac{1}{V(\om,t)}.
\label{Mnt_upbnd2}
\eeq
\label{conj:V}
\end{conjecture}

Note that \eqref{Mnt_upbnd2} is tighter than \eqref{Mnt_upbnd1}, since $|\Phi_t(\x)| \le V(\omega(\x'),t)$.
For $t=1$, the two bounds are identical by virtue of Proposition~\ref{prop:small_t}(a); hence, in this case, 
Theorem~\ref{thm:Mnt_upbnd1} shows that the conjecture is true. We can also prove that the conjecture holds 
for $t=2,3$.

\begin{theorem}
For any positive integer $n$ and $t = 1,2,3$, we have
$$
M(n,t) \le 2 \, \sum_{\om = 0}^{n-1} \binom{n-1}{\om} \frac{1}{V(\om,t)}
$$
\label{thm:Mnt_upbnd2}
\end{theorem}

Table~\ref{table:upbnd2} lists the numerical values of the bound in the above theorem for some small values of $n$.
Again, for the sake of comparison, the corresponding lower bounds from Table~I of \cite{SR11} are given in parentheses. We do not tabulate the row for $t=1$ as this is the same as that in Table~\ref{table:upbnd}.

\begin{table}[ht]
\begin{center}
\begin{tabular}{|c||c|c|c|c|c|c|c|c|c|} \hline
$\overset{\text{\normalsize $\;\;\;n$}}{\underset{\text{\normalsize $t\;\;\;$}}{}}$ & 4 & 5 & 6 & 7 & 8 & 9 & 10 & 15 & 20 \\ \hline\hline
2 & 7  (4) & 10 (8) & 15 (10) & 24 (16) & 39 (22) & 62 & 102 & 1406 & 24306 \\ \hline
3 &  &  &  15 (8) & 23 (16) & 34 (18) & 53 (32) & 81 & 800 & 9921 \\ \hline
\end{tabular}
\end{center}
\caption{Some numerical values of the upper bound on $M(n,t)$ of Theorem~\ref{thm:Mnt_upbnd2}, 
rounded down to the nearest integer.}
\label{table:upbnd2}
\end{table}

In the remainder of this section, we give a proof of Theorem~\ref{thm:Mnt_upbnd2}.

\subsection{Proof for $t=2$}

Fix $\x \in \S^n$, and let $\om = \om(\x')$. We want to prove \eqref{eq:sum_y2} for $t=2$. From the 
discussion in Section~\ref{sec:derivative}, we know that for any $\y \in \Phi_2(\x)$, the Hamming weight of 
$\y'$ must lie between $\om-4$ and $\om$. For $j = 0,1,2,3,4$, let $A_j$ be the number of sequences 
$\y \in \Phi_2(\x)$ such that $\om(\y') = \om-j$. 

\begin{lemma}
Let $m$ denote the number of $1$-runs in $\x'$, and let $m_1$ be the number of these that are of length $1$. 
Then,
\begin{itemize}
\item[(a)] $A_0+A_1 = 1 + m + \binom{m}{2}$;
\item[(b)] $A_2+A_3 = (\om-m)(m+1) - (m-m_1)$;
\item[(c)] $A_4 = \binom{\om-m}{2}-(\om-m) + (m-m_1)$.
\end{itemize}
\label{lem:Aj}
\end{lemma}
\begin{proof} Let $\y = \phi_E(\x)$ for some $E \in \cE_{n,2}$. Write $\x' = (x'_2,\ldots,x'_n)$. Note that
$\x'$ contains $m$ trailing $1$s and $\om-m$ non-trailing $1$s.

(a)\ We have $\om(\y') = \om$ or $\om-1$ iff each $j \in E$ acts upon a trailing $1$ of $\x'$. 
Let $J = \{j \in \{2,\ldots,n\}: x'_j \text{ is a trailing } 1\}$ be the positions of the trailing $1$s in $\x'$.
Thus, $|J| = m$, and $J$ does not contain consecutive integers. The sequence $\y$ is counted by $A_0+A_1$
iff $E \subseteq J$. The number of such grain patterns $E$ is precisely $1 + m + \binom{m}{2}$.

(b)\ We have $\om(\y') = \om-2$ or $\om-3$ iff exactly one $j \in E$ acts upon a non-trailing $1$ in $\x'$. 
Thus, for a grain pattern $E \in \cE_{n,2}$ to contribute to $A_2+A_3$, exactly one grain in the pattern must
act on a non-trailing $1$. The number of such grain patterns $E$ with $|E| = 1$ is precisely $\om - m$. 
It remains to count the number of grain patterns $E$ of cardinality $2$ that contribute to $A_2+A_3$. 
Let $E = \{i,j\}$, where $i$ and $j$ are the grains acting on a trailing $1$ and a non-trailing $1$, respectively. 
If $i$ acts on an ``isolated'' $1$, i.e., a $1$-run of length $1$, then $j$ can act on any of the $\om-m$ 
non-trailing $1$s. On the other hand, if $i$ acts on a trailing $1$ from a $1$-run of length at least $2$, 
then $j$ can be any of the non-trailing $1$s \emph{except} 
for the $1$ at position $i-1$. It follows that the number of grain patterns of cardinality $2$ contributing to 
$A_2+A_3$ equals $m_1(\om-m) + (m-m_1)(\om-m-1)$. Thus,
\beqa
A_2+A_3 &=& (\om-m) + m_1(\om-m) + (m-m_1)(\om-m-1) \\ &=& (\om-m)(m+1) - (m-m_1).
\eeqa

(c)\ This part follows from the fact that $A_4 = |\Phi_2(\x)| - \sum_{j=0}^3A_j$, using the expression for
 $|\Phi_2(\x)|$ given in Proposition~\ref{prop:small_t}(b).
\end{proof}

We are now ready to prove (\ref{eq:sum_y2}). For convenience, we use $V(a)$ to denote $1+a+\binom{a}{2}$.
We start with 
\begin{eqnarray}
\sum_{\y \in \Phi_2(\x)} \frac{1}{V(\om(\y'),2)} 
&\ge& \frac{A_0+A_1}{V(\om)} + \frac{A_2+A_3}{V(\om-2)} + \frac{A_4}{V(\om-4)} \notag \\
&=& 1-\frac{\frac12 (\om-m)(\om+m+1)}{V(\om)} + \frac{A_2+A_3}{V(\om-2)} + \frac{A_4}{V(\om-4)}.
\label{eq:V1}
\end{eqnarray}
The equality above simply uses the fact that $V(\om) - V(m) = \frac12 (\om-m)(\om+m+1)$. Now, note that 
\beqa
\frac{A_2+A_3}{V(\om-2)} + \frac{A_4}{V(\om-4)} 
&\ge& \frac{(\om-m)(m+1)}{V(\om-2)} + \frac{\binom{\om-m}{2}-(\om-m)}{V(\om-4)} \\
&=& \frac{(\om-m)(m+1)}{V(\om-2)} + \frac{\frac12(\om-m)(\om-m-3)}{V(\om-4)}. 
\eeqa
Therefore, carrying on from \eqref{eq:V1}, we have
\begin{eqnarray}
\sum_{\y \in \Phi_2(\x)} \frac{1}{V(\om(\y'),2)} 
&\ge& 1 + (\om-m)\left[\frac{m+1}{V(\om-2)} + \frac{\frac12(\om-m-3)}{V(\om-4)} - \frac{\frac12 (\om+m+1)}{V(\om)} \right] \label{eq:V2} \\
&\ge& 1 + (\om-m)\left[\frac{m+1+\frac12(\om-m-3)}{V(\om-2)} - \frac{\frac12 (\om+m+1)}{V(\om)} \right] \notag \\
&=& 1 + \frac12(\om-m) \left[\frac{\om+m-1}{V(\om-2)} - \frac{\om+m+1}{V(\om)}\right] \notag \\
&=& 1 + \frac{\frac12(\om-m)}{V(\om-2)V(\om)} \, [\om^2+2m\om-(m+3)]. \label{eq:V3}
\end{eqnarray}

If $\om = m$, then (\ref{eq:V2}) proves \eqref{eq:sum_y2}. Else, if $\om \ge m+1$, then the term within square
brackets in \eqref{eq:V3} can be further bounded as follows:
\beqa
\om^2+2m\om-(m+3) &\ge& (m+1)^2 + 2m(m+1) - (m+3) \\
&=& 3m^2+3m - 2,
\eeqa
which is positive for $m \ge 1$. Thus, again, we have \eqref{eq:sum_y2}, which completes the proof of 
the $t=2$ case.

\subsection{Proof for $t=3$}

The approach is the same as that for $t=2$, but the computations are more cumbersome. 
So, let $\x \in \S^n$ be fixed, and let $\om = \om(\x')$. The Hamming weight of $\y'$, for any $\y \in \Phi_3(\x)$,
lies between $\om-6$ and $\om$. For $j = 0,1,\ldots,6$, let $B_j$ be the number of $\y \in \Phi_3(\x)$ such that $\om(\y') = \om-j$.

\begin{lemma}
Let $m$ denote the number of $1$-runs in $\x'$, and let $m_i$, $i = 1,2$, be the number of these that are 
of length $i$. Then,
\begin{itemize}
\item[(a)]  $B_0+B_1 = 1 + m + \binom{m}{2} + \binom{m}{3}$;
\item[(b)]  $B_2+B_3 = (\om-m)\bigl(1+m+\binom{m}{2}\bigr) - m(m-m_1)$;
\item[(c)]  $B_4+B_5 = (1+m)\left[\binom{\om-m}{2} - (\om-m)\right]- (\om-2m-3)(m-m_1) - m_2$;
\item[(d)]  $B_6 = \binom{\om-m}{3}  - (\om-m)(\om-m+1)+ (\om-m)(m-m_1) + 4(\om-2m+m_1) + m_2$.
\end{itemize}
\label{lem:Bj}
\end{lemma}
\begin{proof}
Let $\y = \phi_E(\x)$ for some $E \in \cE_{n,3}$. 

(a)\ This is proved by an easy extension of the proof of Lemma~\ref{lem:Aj}(a).

(b)\ For a grain pattern $E \in \cE_{n,2}$ to contribute to $B_2+B_3$, exactly one grain in the pattern must
act on a non-trailing $1$. The number of such grain patterns $E$ with $|E| \le 2$ is equal to 
$(\om-m)(m+1) - (m-m_1)$ by Lemma~\ref{lem:Aj}(b). Extending the arguments in the proof of 
Lemma~\ref{lem:Aj}(b), we determine that the number of grain patterns of cardinality $3$ that contribute to 
$B_2+B_3$ is equal to $\binom{m_1}{2}(\om-m) + m_1(m-m_1)(\om-m-1) + \binom{m-m_1}{2}(\om-m-2)$.
Thus, 
\beqa
B_2 + B_3 &=& (\om-m)(m+1) - (m-m_1) + \binom{m_1}{2}(\om-m) \\
 & & + \  m_1(m-m_1)(\om-m-1) + \binom{m-m_1}{2}(\om-m-2),
\eeqa
which simplifies to $(\om-m)\bigl(1+m+\binom{m}{2}\bigr) - m(m-m_1)$.

(c)\ This part follows from the fact that $B_4+B_5 = |\Phi_3(\x)| - \sum_{j=0}^3B_j - B_6$, 
using the expression for $|\Phi_3(\x)|$ given in Proposition~\ref{prop:small_t}(c).

(d)\ $B_6$ equals the number of grain patterns $E \in \cE_{n,3}$ with $|E| = 3$, in which all three grains 
act on non-trailing $1$s of $\x'$. The sequence $\x'$ has $m-m_1$ $1$-runs of length at least $2$; let
$\ell_1,\ldots,\ell_{m-m_1}$ denote the lengths of these runs. Then, for $i = 1,\ldots,m-m_1$, $\ell_i^- = \ell_i - 1$
denotes the number of non-trailing $1$s in these runs. With this, we can write
$$
B_6 = \sum_{(i,j,k): i < j< k} \ell^-_i \ell^-_j \ell^-_k + \sum_{i=1}^{m-m_1} \binom{\ell^-_i - 1}{2}\biggl(\sum_{j: j \ne i} \ell^-_j\biggr) + \sum_{i=1}^{m-m_1} \binom{\ell^-_i - 2}{3}. 
$$
From this, straightforward algebraic manipulations yield the expression in the statement of the lemma. 
The algebra here is analogous to that needed to prove Proposition~\ref{prop:small_t}(c).
\end{proof}

For convenience, we define $U(a)$ to be $1+a+\binom{a}{2} + \binom{a}{3}$. We then have
\beq
\sum_{\y \in \Phi_2(\x)} \frac{1}{V(\om(\y'),3)} \ge \sum_{j=0}^2 \frac{B_j + B_{j+1}}{U(\om-2j)} + \frac{B_6}{U(\om-6)}.
\label{eq:U_main}
\eeq
The aim is to show, using Lemma~\ref{lem:Bj}, that the right-hand side of the above inequality is at least $1$.
We dispose of an easy case first. If $\om = m$, then note that we must have $m = m_1 = \om$, and $m_2 = 0$.
With this, Lemma~\ref{lem:Bj} yields $B_0+B_1 = U(\om)$, and $B_2+B_3 = B_4 + B_5 = B_6 = 0$. Hence,
the right-hand side of \eqref{eq:U_main} simplifies to $\frac{U(\om)}{U(\om)} = 1$. This proves the desired 
inequality \eqref{eq:sum_y2} when $\om = m$. 

Also, for small values of $\om$, it can be checked by direct computation using Lemma~\ref{lem:Bj}
that the right-hand side of \eqref{eq:U_main} is at least $1$. We used a computer to check this for $\om \le 16$
and all valid choices of $m$, $m_1$ and $m_2$. Here, ``valid'' means that these quantities must be 
realizable as the number of $1$-runs of the appropriate type in a binary sequence $\x'$ of 
Hamming weight $\om$.

Thus, we may henceforth assume that $1 \le m \le \om-1$ and $\om \ge 17$.

We carry out some more simplifications. The idea is to justify ignoring the terms that involve $m_1$ and $m_2$ 
in the formulae stated in Lemma~\ref{lem:Bj}. When we expand out 
$\frac{B_2+B_3}{U(\om-2)} + \frac{B_4+B_5}{U(\om-4)} + \frac{B_6}{U(\om-6)}$ using Lemma~\ref{lem:Bj}, 
we obtain an expression that includes the following terms:
$$
-\frac{m(m-m_1)}{U(\om-2)} - \frac{(\om-2m-3)(m-m_1) + m_2}{U(\om-4)} 
+ \frac{(\om-m)(m-m_1) + 4(\om-2m+m_1) + m_2}{U(\om-6)}
$$
Re-write this as 
\begin{align*}
m(m-m_1)\left[\frac{1}{U(\om-4)} - \frac{1}{U(\om-2)}\right] 
&+ [(\om-m)(m-m_1)+m_2] \, \left[\frac{1}{U(\om-6)} - \frac{1}{U(\om-4)}\right] \\
&+ \, \frac{3(m-m_1)}{U(\om-4)} + \frac{4(\om-2m+m_1)}{U(\om-6)}.
\end{align*}
The above expression is a sum of four terms, each of which is non-negative. (To see that the last term is
non-negative, observe that $\om \ge m_1 + 2(m-m_1) = 2m - m_1$; this is because each $1$-run counted by 
$m_1$ contains exactly one $1$, while the remaining $m-m_1$ $1$-runs contain at least two $1$s each.) 
Therefore, the sum $\frac{B_2+B_3}{U(\om-2)} + \frac{B_4+B_5}{U(\om-4)} + \frac{B_6}{U(\om-6)}$ is at least
$$
\frac{(\om-m)(1+m+\binom{m}{2})}{U(\om-2)} 
+ \frac{(1+m)\left[\binom{\om-m}{2} - (\om-m)\right]}{U(\om-4)} 
+ \frac{\binom{\om-m}{3}  - (\om-m)(\om-m+1)}{U(\om-6)},
$$
which can also be expressed as 
\beq
(\om-m) \, 
\left[ \frac{1+m+\binom{m}{2}}{U(\om-2)} + \frac{\frac12 (1+m)(\om-m-3)}{U(\om-4)}
+ \frac{\frac16 \bigl[(\om-m)^2 - 9(\om-m) - 4\bigr]}{U(\om-6)}
\right].
\label{sum_b-d}
\eeq
Next, we write 
\begin{eqnarray}
\frac{B_0 + B_1}{U(\om)} &=& \frac{U(m)}{U(\om)}  \ \ = \ \ 1 - \frac{U(\om) - U(m)}{U(\om)} \notag \\
& = & 1 - \frac{\frac16 (\om-m) (\om^2+\om m + m^2 + 5)}{U(\om)}. 
\label{sum_a}
\end{eqnarray}

Putting \eqref{sum_b-d} and \eqref{sum_a} together, we find that the right-hand side of \eqref{eq:U_main}
is lower bounded by 
\beq
1 + (\om-m) g_\om(m), 
\label{eq:g_om}
\eeq
where 
\beqa
g_\om(m) & = & \frac{1+m+\binom{m}{2}}{U(\om-2)} + \frac{\frac12 (1+m)(\om-m-3)}{U(\om-4)} \\ 
& & + \, \frac{\frac16 \bigl[(\om-m)^2 - 9(\om-m) - 4\bigr]}{U(\om-6)} - \frac{\frac16(\om^2+\om m + m^2 + 5)}{U(\om)}. 
\eeqa

For a fixed $\om$, consider $g_\om$ as a function of $m$. Some tedious computations (some of which 
were performed with the aid of Maple) show the following:
\begin{itemize}
\item for $\om \ge 6$, $g_\om$ is a convex function, i.e., $g_\om''(x) \ge 0$ for $1 \le x \le \om$;
\item for $\om \ge 12$, $g_\om'(\om-1) \le 0$;
\item for $\om \ge 13$, $g_\om(\om-1) \ge 0$.
\end{itemize}
From this, we obtain the fact that, as long as $\om \ge 13$, we have $g_\om(m) \ge 0$
for $1 \le m \le \om-1$. Thus, for these values of $\om$ and $m$, \eqref{eq:g_om} yields that the 
right-hand side of \eqref{eq:U_main} is lower bounded by $1$. Recalling that we only needed to show
this for $\om \ge 17$, the proof of the $t=3$ case in Theorem~\ref{thm:Mnt_upbnd2} is complete.

\section{An Upper Bound on $R(\tau)$\label{sec:Rtau_upbnd1}}

Were they to be proved, Conjectures~\ref{conj:Phi_t} and \ref{conj:V} would yield upper bounds on the asymptotic rate $R(\tau)$, as defined in \eqref{eq:Rtau}. Instead, a slightly different approach\footnote{This approach was suggested to the authors by Artyom Sharov and Ronny Roth.} can be used to obtain a fractional covering that does result in a provable upper bound on $R(\tau)$. 

Suppose that for any fixed $n,t$, we could find a lower bound $\varphi_{n,t}(r)$ on $|\Phi_t(\x)|$, $\x \in \S^n$, that depends on $\x$ only through $r = r(\x)$, the number of distinct runs in $\x$. Furthermore, suppose that the function $\varphi_{n,t}(r)$ is non-decreasing in $r$ \cite[Section~3]{SR13}. 
Then, it is straightforward to see that the function $\x \mapsto \frac{1}{\varphi_{n,t}(r(\x))}$ is a fractional covering of the hypergraph $\cH_{n,t}$ for all positive integers $n$ and $t$.  Indeed, by Lemma~\ref{lem:runs}, we have
$$
\sum_{\y \in \Phi_t(\x)} \frac{1}{\varphi_{n,t}(r(\y))} \ge \sum_{\y \in \Phi_t(\x)} \frac{1}{\varphi_{n,t}(r(\x))} = 
\frac{|\Phi_t(\x)|}{\varphi_{n,t}(r(\x))} \ge 1.
$$
Thus, for any such $\varphi_{n,t}$, we have
\beq
M(n,t) \le \sum_{\x \in \Phi_t(\x)} \frac{1}{\varphi_{n,t}(r(\x))} = 2 \, \sum_{r=1}^n \binom{n-1}{r-1} \frac{1}{\varphi_{n,t}(r)}.
\label{eq:Mnt_upbnd3}
\eeq

\begin{theorem}
For all positive integers $n$ and $t$, the upper bound \eqref{eq:Mnt_upbnd3} holds with 
\beq
\varphi_{n,t}(r) = \sum_{j=0}^t \binom{r-j}{j}
\label{eq:varphi}
\eeq
\label{thm:Mnt_upbnd3}
\end{theorem}
\begin{proof}
The expression on the right-hand side of \eqref{eq:varphi} is clearly non-decreasing in $r$, and by Proposition~\ref{prop:Phi_lobnd}, $\varphi_{n,t}(r(\x))$ is a lower bound on $|\Phi_t(\x)|$ for any $\x \in \S^n$ and $t > 0$. 
\end{proof}

The bound of Theorem~\ref{thm:Mnt_upbnd3} is weaker than that of Theorems~\ref{thm:Mnt_upbnd1} and \ref{thm:Mnt_upbnd2} for $t = 1,2,3$. However, it has the advantage of being provably true for all values of $n$ and $t$. It can therefore be used to derive an upper bound on $R(\tau)$ by studying the asymptotics of $\varphi_{n,t}(r)$ as $n \to \infty$, with $t = \lceil\tau n\rceil$ and $r = \lceil\rho n\rceil$ for fixed $\tau \in [0,\frac12]$ and $\rho \in (0,1]$. The following theorem is proved in Appendix~B.

\begin{theorem}
Let $\phi = \frac{1+\sqrt{5}}{2}$ (the golden ratio), and define $\theta = \frac{1}{\sqrt{5}\phi(\phi+1)}$. For $\tau \in [0,\frac12]$, we have
$$
R(\tau) \le 
\begin{cases} 
{\displaystyle \max_{\sqrt{5}\phi\tau \le \rho \le 1}} \ \left[\sfh(\rho) - (\rho-\tau) \, \sfh\left(\frac{\tau}{\rho-\tau}\right)\right] & \text{ if } \tau < \theta \\ 
\log_2\phi & \text{ if } \tau \ge \theta
\end{cases}
$$
Numerically,  $\theta \approx 0.1056$, and $\log_2 \phi  \approx 0.6942$.
\label{thm:Rupbnd1}
\end{theorem}

Figure~\ref{fig:upper} contains a plot of the above upper bound. The figure shows that this is the best known upper bound for values of $\tau$ up to about $0.1103$, beyond which it is beaten by the bound of the next section.

\section{An Information-Theoretic Upper Bound on $R(\tau)$\label{sec:Rtau_upbnd2}}

In this section, we use an information-theoretic approach to derive an upper bound on $R(\tau)$. 
For every even $n$, by grouping together adjacent coordinates, we can view any code $C\in\{0,1\}^n$ as 
a code of blocklength $n/2$ over the alphabet $\{00,01,10,11\}$. Let us say that a binary $n$-tuple, 
alternatively an $n/2$-tuple over the quaternary alphabet, has {\em quaternary distribution} 
(or simply {\em distribution})
  $(f_{00},f_{11},f_{01},f_{10})$
if it has $f_{00}n/2$ symbols $00$, $f_{11}n/2$ symbols $11$, $f_{01}n/2$ symbols $01$ and $f_{10}n/2$
symbols $10$. We will say that a code has {\em constant distribution} if
each of its codewords has the same quaternary distribution
$(f_{00},f_{11},f_{01},f_{10})$.
Our goal is to find upper bounds on the rate of $\lceil \tau n\rceil$-grain-correcting codes of constant distribution: 
since the number of possible quaternary distributions for a code of length $n$ is $O(n^3)$, the maximum
of these upper bounds on constrained codes will yield an unconstrained upper bound.

Let us introduce the following notation:
$$
R_f(\tau) = \limsup_{n\to\infty} \frac{1}{n} \, \log_2 M(n,f,\lceil \tau n \rceil)
$$
where $M(n,f,t)$ denotes the maximum cardinality of a $t$-grain error
correcting code of length $n$ and constant quaternary distribution $f$.

Our strategy is the following: for any given distribution
$f = (f_{00},f_{11},f_{01},f_{10})$, we associate to
it a discrete memoryless channel (DMC) with input and output alphabets
$\{00,01,10,11\}$
such that any infinite family of $\lceil\tau n\rceil$-grain-correcting codes of constant distribution $f$ 
achieves vanishing error-probability when submitted through this channel. By a standard information-theoretic
argument, this implies that the asymptotic rate $R$ of any family of $\lceil\tau n\rceil$-grain-correcting codes of
constant distribution $f$ is bounded from above by half the mutual information between the channel input
with probability distribution $f$ and the channel output.

\begin{figure}[t]
  \centering
\begin{tikzpicture}
    \path (0,6) node (in) {In} (3,6) node (out) {Out}
          (0,5) node (x11) {$11$} (3,5) node (y11) {$11$}
          (0,3) node (x10) {$10$} (3,3) node (y10) {$10$}
          (0,2) node (x00) {$00$} (3,2) node (y00) {$00$}
          (0,0) node (x01) {$01$} (3,0) node (y01) {$01$};
    \draw[->] (x01) -- node[above] {$1-p$} (y01);
    \draw[->] (x00) -- (y00);
    \draw[->] (x01) -- node[above,sloped] {$p$} (y00);

    \draw[->] (x10) -- node[above] {$1-p$} (y10);
    \draw[->] (x11) -- (y11);
    \draw[->] (x10) -- node[above,sloped] {$p$} (y11);   
  \end{tikzpicture}
    \caption{A DMC whose effect can be mimicked by grain patterns}
    \label{fig:dmc}
\end{figure}
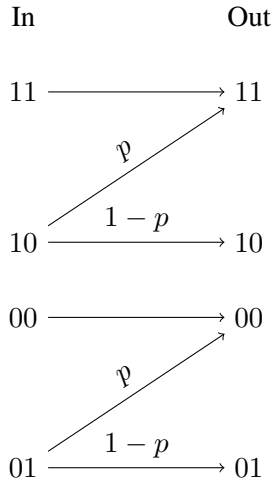

Consider the channel depicted in Figure~\ref{fig:dmc}. Let $C$ be a member of a family of 
$\lceil\tau n\rceil$-grain-correcting codes of length $n$ and constant distribution $f$. Suppose that 
$$(f_{10}+f_{01})pn/2 \leq \tau n(1-\varepsilon),$$
where $p$ is the transition probability shown in Figure~\ref{fig:dmc}.
When a binary $n$-tuple, equivalently a word of length $n/2$ over the
alphabet $\{00,01,10,11\}$, is transmitted over the channel, then with
probability tending to $1$ as $n$ goes to infinity, the number of 
transitions $01\rightarrow 00$ plus the number of transitions
$10\rightarrow 11$ is not more than $\lceil\tau n\rceil$. Since these
transitions are of the kind caused by grain errors, if there are no more
than $\lceil\tau n\rceil$ such transitions, then the errors they cause
are correctable by any $\lceil\tau n\rceil$-grain-correcting code. 
Therefore, for any $\varepsilon >0$, any family of $\lceil\tau n\rceil$-grain-correcting codes 
of constant distribution $f$ can be transmitted over the above channel with vanishing error
probability after decoding. By a continuity argument we conclude that:
\begin{equation}
  \label{eq:Rf}
  R_f(\tau) \leq \frac 12 I(X,Y)
\end{equation}
where $X$ is the channel input with probability distribution $p(X)=f$, 
and $Y$ is the corresponding output of the channel with parameter
\begin{equation}
  \label{eq:p}
  p=\frac{2\tau}{f_{10}+f_{01}}.
\end{equation}

\medskip

It remains to compute the mutual information $I(X,Y)$. Since $p\leq 1$, \eqref{eq:p}
implies that we can write
\begin{align}
  f_{10}+f_{10} &= 2\tau +x\label{eq:2t+x}\\
  f_{00}+f_{11}&= 1-2\tau -x\label{eq:1-2t-x}
\end{align}
with $x$ non-negative. Now, for every distribution satisfying
\eqref{eq:2t+x} and \eqref{eq:1-2t-x} we have
$$H(Y\,|\, X) = \left(2\tau +
    x\right)\sfh\left(\frac{2\tau}{2\tau +x}\right),$$
where $\sfh(\cdot)$ is the binary entropy function defined by $\sfh(\xi) = -\xi\log_2\xi - (1-\xi)\log_2(1-\xi)$, 
for $\xi \in [0,1]$. This implies that $I(X,Y)=H(Y)-H(Y\,|\, X)$ is maximum under the 
constraints \eqref{eq:2t+x} and \eqref{eq:1-2t-x} when $H(Y)$ is
maximum, i.e. under the distribution:
    $$P(Y=10)=P(Y=01)=\frac x2,\;\;P(Y=00)=P(Y=11)=\frac{1-x}{2}.$$
Therefore, we obtain
\begin{equation}
  \label{eq:I(XY)}
  I(X,Y) \leq 1+\sfh(x) - \left(2\tau +
    x\right)\sfh\left(\frac{2\tau}{2\tau +x}\right),
\end{equation}
which together with \eqref{eq:Rf} gives
  $$R_f(\tau) \leq \frac 12 \left[1+\sfh(f_{10}+f_{01}-2\tau) -
    \left(f_{10}+f_{01}\right)\sfh\left(\frac{2\tau}{f_{10}+f_{01}}\right)\right]. $$
The right hand side of \eqref{eq:I(XY)} is maximized for $x=1/2-\tau$,
thus yielding the unconstrained upper bound stated in the theorem below.

\begin{theorem} For $\tau \in [0,\frac12]$, we have
$$
  R(\tau) \leq \frac 12\left( 1+ \sfh\left(\frac 12-\tau\right) -
    \left(\frac 12 +\tau\right)\sfh\left(\frac{2\tau}{\frac 12 +
        \tau}\right)\right).
$$
\label{thm:Rupbnd2}
\end{theorem}

\begin{figure}[th]
\centerline{\scalebox{0.6}{\includegraphics{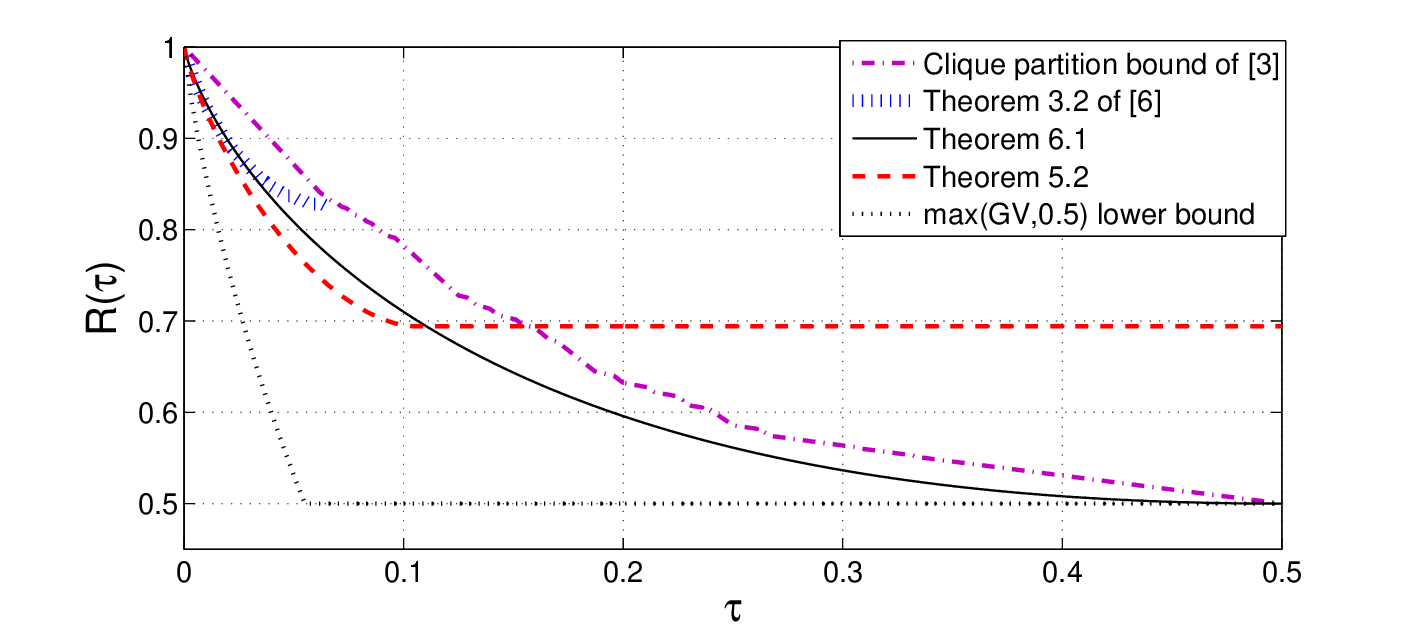}}}
\caption{The upper bounds of Theorems~\ref{thm:Rupbnd1} and \ref{thm:Rupbnd2}, along with bounds from \cite{MBK11} and \cite{SR13}.}
\label{fig:upper}
\end{figure}

The upper bounds of Theorems~\ref{thm:Rupbnd1} and \ref{thm:Rupbnd2} are plotted in Figure~\ref{fig:upper}. 
For comparison, also plotted are upper and lower bounds from \cite[Figure~1]{MBK11}, and the upper bound of Sharov and Roth \cite[Theorem~3.2]{SR13}. The upper bounds from \cite{MBK11} and \cite{SR13} are the best bounds in the prior literature. Figure~\ref{fig:upper} clearly shows that the upper bounds of Theorem~\ref{thm:Rupbnd1} and \ref{thm:Rupbnd2}  improve upon the previously known upper bounds, but still remain far from the lower bound plotted. It should be pointed that a slightly better lower bound was found by Sharov and Roth \cite{SR11}. Unfortunately, the improvement is only minor: the lower bound of \cite{SR11} remains above $0.5$ only in the interval $[0,0.0566]$, and
in that interval, the improvement does not exceed $0.012$. 



\section{Concluding Remarks\label{sec:conclusion}}

In this paper, we derived upper bounds on the maximum cardinality, $M(n,t)$, of a binary 
$t$-grain-correcting code of blocklength $n$, and also on the asymptotic rate $R(\tau$). In nearly all cases, 
the gap between the upper bound and the best known lower bound remains significant. A natural 
question to ask is whether the putative upper bounds on $M(n,t)$ in Conjectures~\ref{conj:Phi_t} and \ref{conj:V}
would yield a better bound on $R(\tau)$. 

\begin{figure}[th]
\centerline{\scalebox{0.6}{\includegraphics{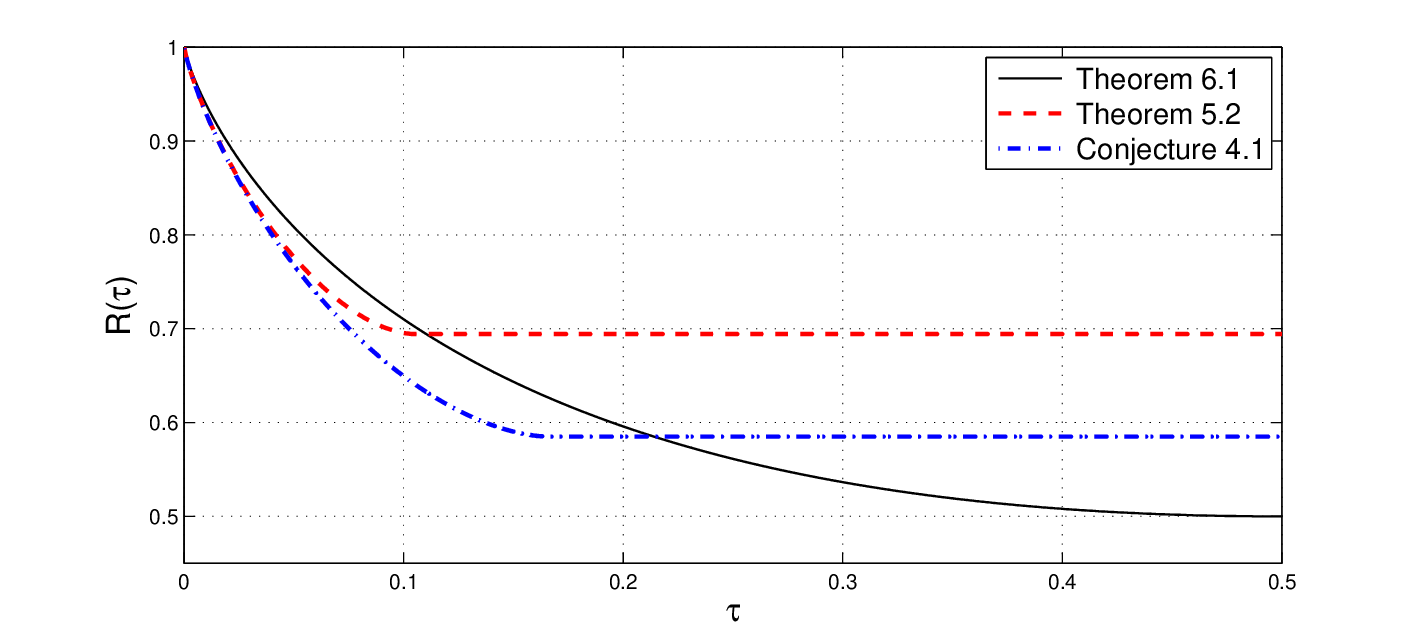}}}
\caption{The upper bounds of Theorems~\ref{thm:Rupbnd1} and \ref{thm:Rupbnd2} compared with the asymptotic bound obtained from Conjecture~\ref{conj:V}.}
\label{fig:conj4.1}
\end{figure}

The bound in Conjecture~\ref{conj:V} is the stronger of the two conjectured bounds, and its asymptotics are 
straightforward to analyze. Let $\overline{M}(n,t) = 2\sum_{\om = 0}^{n-1} \binom{n-1}{\om} \, \frac{1}{V(\om,t)}$ 
denote the upper bound in (\ref{Mnt_upbnd2}), and let 
$\overline{R}(\tau) = \lim_{n\to\infty} \frac{1}{n} \log_2 \overline{M}(n,n\tau)$. Conjecture~\ref{conj:V} implies that
$R(\tau) \le \oR(\tau)$. By standard asymptotic analysis, we obtain
$$\oR(\tau) = \max_{0 \le \nu \le 1} [\sfh(\nu) - \eta(\nu)],$$
where $\eta(\nu)$ equals $\nu$ if $\nu \le 2\tau$, and equals $\nu \sfh(\tau/\nu)$ otherwise. Thus,
$$
\oR(\tau) = \max \biggl\{\max_{0 \le \nu \le 2\tau} [\sfh(\nu) - \nu], 
    \, \max_{2\tau \le \nu \le 1} [\sfh(\nu) - \nu \sfh(\tau/\nu)]\biggr\}.
$$

Using elementary calculus to solve the two maximization problems within the braces in the above equation, and 
comparing the solutions (details of these calculations are omitted), we obtain the following:
\beq
\oR(\tau) = \begin{cases}
\sfh(\nu^*) - \nu^*\sfh(\tau/\nu^*) & \text{ if } \tau < 1/6 \\
\sfh(1/3) - 1/3 & \text{ if } \tau \ge 1/6
\end{cases}
\label{eq:oR1}
\eeq
where $\nu^* = \frac{1}{4}(\tau+1+\sqrt{\tau^2-6\tau+1})$. Numerically, $\sfh(1/3)-1/3 \approx 0.5850$. This bound is compared with the bounds of 
Theorem~\ref{thm:Rupbnd2} in Figure~\ref{fig:conj4.1}. The plot shows that the conjectured upper bound 
(\ref{eq:oR1}) is (expectedly) better than the bound of Theorem~\ref{thm:Rupbnd1}, and improves upon the bound of Theorem~\ref{thm:Rupbnd2} for $\tau < 0.214$.

\section*{Appendix~A}
 \renewcommand{\thetheorem}{A.\arabic{theorem}} 

In this appendix, we prove the following lemma.

\begin{lemma}
For $t=2,3$, a pairing satisfying (\ref{pairing}) can be constructed.
\label{lemma:pairing}
\end{lemma}

We introduce some convenient notation that will be used in the proof. Notation of the form 
$$
\Breve{a} b \longrightarrow c d \ \  \text{ or }  \ \ \Breve{a} b \stackrel{E}{\longrightarrow} c d 
$$
will be used to denote the fact that a length-2 grain (from the grain pattern $E$) 
acting on $a$ converts the substring $ab$ of the derivative sequence $\x'$ to the substring $cd$ of $\y'$.

\medskip

\underline{Case I: $t=2$}.\ \  
Let $\y \in \Phi_2(\x)$. Let $\om$ and $m$ be the number of $1$-runs and the Hamming
weight, respectively, of $\x'$, and let $\tom$ and $\tm$ denote the same for $\y'$. Since the action of a grain
pattern cannot increase the weight of $\x'$, we have $\tom \le \om$. Also, note that a single grain
can cause the number of $1$-runs in $\x'$ to increase by a most $1$ --- an increase by $1$ happens 
either when $*1\Breve{1}00* \longrightarrow *1010*$, or when $*1\Breve{1}11* \longrightarrow *1001*$.
Thus, $\tm \le m+2$. 

We first show that $|\Phi_2(\y)| > |\Phi_2(\x)|$ iff $\tom = \om$ and $\tm > m$. The ``if'' part is clearly true
by Proposition~\ref{prop:small_t}(b). For the ``only if'' part, suppose that $\tm \le m$. Then, 
since $\tom \le \om$ is always true, we have $|\Phi_2(\y)| \le |\Phi_2(\x)|$ by Proposition~\ref{prop:small_t}(b). 
On the other hand, if $\tom \le \om -1$, then since $\tm \le \tom$ is always true (the number of $1$-runs
cannot exceed the number of $1$s), we have 
$$
|\Phi_2(\y)| - |\Phi_2(\x)| = (\tm - m) + \binom{\tom}{2} - \binom{\om}{2} 
\le \tom + \binom{\om-1}{2} - \binom{\om}{2}
= \tom - (\om-1) \le 0.
$$
Hence, $|\Phi_2(\y)| \le |\Phi_2(\x)|$.

Thus, $F_2(\x) = \{\y \in \Phi_2(\x): \tom = \om \text{ and } \tm > m\}$. If $\x$ has weight $\om = 1$, 
then there is no $\y \in \Phi_2(\x)$ for which $\tm > m$, so that $F_2(\x) = \emptyset$. We henceforth 
consider $\om \ge 2$. We assume $\y \in F_2(\x)$, and suppose that $E \in \cE_{n,2}$ is a grain pattern 
such that $\y = \phi_E(\x)$. Since $\tom = \om$, the grains
in $E$ do not act upon non-trailing $1$s in $\x'$. Let $\tm = m+a$, $a = 1$ or $2$, 
so that $|\Phi_2(\y)| = |\Phi_2(\x)| + a$. We will construct a $\z \in G_2(\x)$ with which $\y$ can be paired.

Suppose first that $a=1$, i.e., $\tm = m+1$. There is a unique segment $*1100*$ of $\x'$ such 
that $*1\Breve{1}00* \stackrel{E}{\longrightarrow} *1010*$. Let $j$ be the position of the trailing $1$ affected.
Consider the grain pattern $E' \in \cE_{n,2}$ that acts instead on the preceding $1$, i.e., the position $j$ in $E$ 
is replaced by ${j-1}$ in $E'$. Let $\z = \phi_{E'}(\x)$, and note that in $\x' \stackrel{E'}{\longrightarrow} \z'$,
the same segment $*1100*$ of $\x'$ now undergoes the change 
$*\Breve{1}100* \stackrel{E'}{\longrightarrow} *0000*$.
Thus, the number of $1$-runs in $\z'$ does not exceed the number, $m$, of $1$-runs in $\x'$; and moreover,
$\om(\z') = \om - 2$. Hence, $|\Phi_2(\z)| - |\Phi_2(\x)| \le \binom{\om-2}{2} - \binom{\om}{2} = 
-(2\om - 3) \le -1$, since we assumed $\om \ge 2$ at the outset. 

Thus, we have $|\Phi_2(\y)| = |\Phi_2(\x)| + 1$ and $|\Phi_2(\z)| \le |\Phi_2(\x)| - 1$. With this, we have
$\frac{1}{|\Phi_2(\y)|}  + \frac{1}{|\Phi_2(\z)|} \ge \frac{2}{|\Phi_2(\x)|}$. So, we can pair $\y \in F_2(\x)$ with $\z$. 

Now, suppose that $a = 2$, i.e., $\tm = m+2$. There are now exactly two segments of $\x'$ such that 
$*1\Breve{1}00* \stackrel{E}{\longrightarrow} *1010*$. Let $E' = E-1$ be obtained by replacing each $j \in E$
by $j-1$, and consider $\z = \phi_{E'}(\x)$. Once again, the number of $1$-runs in $\z'$ does not exceed $m$,
but now, we have $\om(\z') = \om-4$. This time, $|\Phi_2(\z)| - |\Phi_2(\x)| 
\le \binom{\om-4}{2} - \binom{\om}{2} = -(4\om - 10)$. Note that $\om$ must be at least $3$, since $\tm = m+2$
is not possible when $\om = 2$. Therefore, $|\Phi_2(\z)| - |\Phi_2(\x)| \le -4(\om-10) \le -2$. 
Thus, $|\Phi_2(\y)| = |\Phi_2(\x)| + 2$ and $|\Phi_2(\z)| \le |\Phi_2(\x)| - 2$. 
Hence, $\frac{1}{|\Phi_2(\y)|}  + \frac{1}{|\Phi_2(\z)|} \ge \frac{2}{|\Phi_2(\x)|}$, and we can pair $\y$ with this $\z$. 

By construction, the pairing $\y \mapsto \z$ is a one-to-one map.

\bigskip

\underline{Case II: $t=3$}.\ \  Consider any $\y \in F_3(\x)$. To the notation introduced above, we add
$m_1$ and $\tm_1$ to denote the number of $1$-runs of length $1$ in $\x'$ and $\y'$, respectively. As before,
$\tom \le \om$, but this time, $\tm \le m+3$ since $\x'$ can be affected by up to three grains. Also, a single grain 
can cause an increase of $2$ in $m_1$: $*01\Breve{1}110* \longrightarrow *010010*$ or 
$*01\Breve{1}00* \longrightarrow *01010*$. Hence, $\tm_1 \le m_1+6$.

Suppose first that $\tom \le \om-2$. Then, from Proposition~\ref{prop:small_t}(c), we have
$|\Phi_3(\y)| \le 1 + (m_1+6) + (m+3)(\om-2-3) + 2(\binom{\om-2}{3} + \binom{\om-2}{2} + 2(\om-2)$. 
Upon simplifying, we obtain $|\Phi_3(\y)| - |\Phi_3(\x)| \le -(\om^2-9\om+20)- 2m$. Since $m \ge 1$, 
we further obtain $|\Phi_3(\y)| - |\Phi_3(\x)| \le -(\om^2-9\om+22)$, which is a negative quantity. 
This cannot happen for $\y \in F_3(\x)$, so $\tom$ must equal $\om-1$ or $\om$. 

Since $\tom$ equals $\om -1$ or $\om$, we may assume that $\y = \phi_E(\x)$, 
for some grain pattern $E \in \cE_{n,3}$ that acts only upon the trailing $1$s in $\x'$. 
Let $E_i$, $i=1,2,3$ be the subset of $E$ consisting of grains $j$ that act on the trailing $1$s of $1$-runs 
of length $i$; also let $E_4 = E \setminus (E_1 \cup E_2 \cup E_3)$ be the subset of $E$ acting on the trailing $1$s 
of $1$-runs of length at least $4$. Set $e_i = |E_i|$, $i = 1,2,3,4$. It is easy to see that $\tm \le m + e_2+e_3+e_4$,
while $\tm_1 \le m_1 + 2e_2+e_3+e_4$. 

Let $d = e_2+e_3+e_4$. If $d=0$, then $\tm \le m$ and $\tm_1 \le m_1$. Since $\tom \le \om$ always, we have
$|\Phi_3(\y)| \le |\Phi_3(\x)|$ by Proposition~\ref{prop:small_t}(c), which is not possible for $\y \in F_3(\x)$. 

At this point, we have that $\tom$ equals $\om -1$ or $\om$, and $d$ equals 1, 2, or 3. 
We will now construct a $\z$ to be paired with $\y$. Let $E' = E_1 \cup \{j-1: j \in E_2 \cup E_3 \cup E_4\}$.
Thus, $E'$ is a grain pattern in $\cE_{n,3}$ that retains the grains from $E$ that act on trailing $1$s from 
$1$-runs of length $1$, but pushes back all the other grains in $E$ by one position. Let $\z = \phi_{E'}(\x)$.
We claim that the desired inequality $\frac{1}{|\Phi_3(\y)|}  + \frac{1}{|\Phi_3(\z)|} \ge \frac{2}{|\Phi_3(\x)|}$ holds. 
The remainder of this proof justifies this claim. 

It is enough to show that $\frac12 [|\Phi_3(\y)| + |\Phi_3(\z)|] \le |\Phi_3(\x)|$, since by the concavity of the function 
$f(x) = 1/x$, we would then have
$$
\frac12 \left[\frac{1}{|\Phi_3(\y)|}  + \frac{1}{|\Phi_3(\z)|}\right] \ge \frac{1}{\frac12 [|\Phi_3(\y)| + |\Phi_3(\z)|]}
\ge \frac{1}{|\Phi_3(\x)|}.
$$
To this end, note first that 
\beq
|\Phi_3(\y)| - |\Phi_3(\x)| \le 2e_2+e_3+e_4 + (\om-3)d.
\label{eq:diffyx}
\eeq
Next, we bound $|\Phi_3(\z)| - |\Phi_3(\x)|$. 
It is not difficult to check that $\om(\z') = \om - 2d$, the number of $1$-runs in $\z'$ is at most $m$, 
and at most $m_1+e_3$ of these are of length $1$. Thus, 
$|\Phi_3(\z)| \le 1 + (m_1+e_3) + m(\om-2d-3) +\binom{\om-2d}{3} - \binom{\om-2d}{2} + 2(\om-2d)$, 
and hence, 
\beq
|\Phi_3(\z)| - |\Phi_3(\x)| \le e_3 - 2md - \xi(\om,d),
\label{eq:diffzx}
\eeq
where $\xi(\om,d) = \binom{\om}{3} - \binom{\om}{2} + 2\om - \bigl[\binom{\om-2d}{3} - \binom{\om-2d}{2} + 2(\om-2d)\bigr]$. 
From (\ref{eq:diffyx}) and (\ref{eq:diffzx}), we obtain
\begin{eqnarray}
\frac12 [|\Phi_3(\y)| + |\Phi_3(\z)|] - |\Phi_3(\x)| 
&\le& e_2+e_3+\frac12 e_4 + \frac12 (\om-3)d - md - \frac12 \xi(\om,d) \notag \\
&\le& d + \frac12 (\om-3) d - d^2 - \frac12 \xi(\om,d), \label{eq:avg}
\end{eqnarray}
where we have used the fact that $m \ge d$, which is simply the observation that $|E_2 \cup E_3 \cup E_4|$ cannot
exceed the number of $1$-runs in $\x'$. 

If $d=1$, the expression in (\ref{eq:avg}) reduces to $-\frac12 (\om^2 - 7\om + 14)$, a negative quantity.
If $d=2$, we obtain $-(\om^2-9\om+24)$ instead, which is still a negative quantity. If $d=3$, we get
$-\frac32 (\om^2-11\om + \frac{110}{3})$, which is also a negative quantity. Thus, in all cases, we have
$\frac12 [|\Phi_3(\y)| + |\Phi_3(\z)|] \le |\Phi_3(\x)|$ as desired.  
\qed

\section*{Appendix B}
 \renewcommand{\thetheorem}{B.\arabic{theorem}} 
 \setcounter{theorem}{0}

We prove Theorem~\ref{thm:Rupbnd1} here. Throughout this appendix, we set $t = \lceil \tau n \rceil$ and $r = \lceil \rho n \rceil$ for some $\tau \in [0,\frac12]$ and $\rho \in [0,1]$.

The asymptotics of $\varphi_{n,t}(r)$ is determined by the largest term within the summation in \eqref{eq:varphi}. Letting $\zeta_j = \binom{r-j}{j}$, it is easy to verify that the ratio $\zeta_{j-1}/\zeta_j$ is at most $1$ when $j \le \frac{1}{10} (5r+7-\sqrt{5r^2+10r+9})$, and is strictly larger than $1$ for $\frac{1}{10} (5r+7-\sqrt{5r^2+10r+9}) < j \le r/2$; for $j > r/2$, we have $\zeta_j = 0$. Therefore, setting $J = \lfloor{\frac{1}{10} (5r+7-\sqrt{5r^2+10r+9})}\rfloor$, we see that if $t < J$, then the dominant term in \eqref{eq:varphi} is $\zeta_t$; and if $t \ge J$, the dominant term is $\zeta_J$. Passing to asymptotics, it follows that if we define $\alpha =\frac{5-\sqrt{5}}{10}$, then 
\beq
\lim_{n \to \infty} \frac{1}{n} \log_2 \varphi_{n,\lceil{\tau n \rceil}}(\lceil \rho n \rceil) = 
\begin{cases}
(\rho-\tau)\sfh\left(\frac{\tau}{\rho-\tau}\right) & \text{ if } \tau \le \alpha \rho \\
\rho (1-\alpha) \sfh(\frac{\alpha}{1-\alpha}) &  \text{ if } \tau \ge \alpha \rho
\end{cases}
\label{eq:lim}
\eeq

We record in the following lemma some facts about the constant $\alpha = \frac{5-\sqrt{5}}{10}$ that will be useful in the sequel. They are proved by straightforward algebraic manipulations. For ease of verification, we give a proof of part (c) at the end of this appendix.

\begin{lemma} Recall that $\phi = \frac{1+\sqrt{5}}{2}$ is the golden ratio. 
\begin{itemize}
\item[(a)] $\alpha^{-1} = \sqrt{5} \phi$ 
\item[(b)] $\frac{\alpha}{1-\alpha} = \frac{1}{1+\phi}$
\item[(c)] $(1-\alpha) \sfh(\frac{\alpha}{1-\alpha}) = \log_2\phi$.
\end{itemize}
\label{lem:alpha}
\end{lemma}

Resuming the proof of Theorem~\ref{thm:Rupbnd1}, from \eqref{eq:Mnt_upbnd3}, we obtain
$$
R(\tau) \le \max_{0 \le \rho \le 1} \left[\sfh(\rho) - \lim_{n\to\infty} \frac{1}{n} \varphi_{n,\lceil \tau n \rceil}(\lceil \rho n \rceil)\right].
$$
Hence, using \eqref{eq:lim} and Lemma~\ref{lem:alpha}, we have
\beq
R(\tau) \le \max\{A(\tau),B(\tau)\},
\label{eq:Rtau1}
\eeq
where 
\beq
A(\tau) = \max_{0 \le \rho \le \min\{\alpha^{-1}\tau,1\}} [\sfh(\rho) - \rho \log_2\phi],
\label{eq:A1}
\eeq
and 
\beq
B(\tau) = \max_{\alpha^{-1}\tau \le \rho \le 1} \left[\sfh(\rho) - (\rho-\tau)\sfh\left(\frac{\tau}{\rho-\tau}\right)\right]
\label{eq:B}
\eeq
For convenience, we define $B(\tau) = 0$ if $\alpha^{-1}\tau > 1$. Note that the term within square brackets in \eqref{eq:B} reduces to $\sfh(\alpha^{-1}\tau) - \alpha^{-1}\tau\log_2\phi$ if we set $\rho = \alpha^{-1}\tau$; therefore, $\sfh(\alpha^{-1}\tau) - \alpha^{-1}\tau \log_2\phi \le B(\tau)$. 

Now, using elementary calculus to solve the maximization problem in \eqref{eq:A1}, we obtain
$$
A(\tau) = \begin{cases}
\sfh(\alpha^{-1}\tau) - \alpha^{-1}\tau\log_2\phi & \text{ if } \alpha^{-1} \tau \le \frac{1}{1+\phi}\\
\sfh\left(\frac{1}{1+\phi}\right) - \frac{1}{1+\phi} \cdot \log_2\phi & \text{ if } \alpha^{-1} \tau \ge \frac{1}{1+\phi}
\end{cases}
$$
Somewhat miraculously, the expression $\sfh(\frac{1}{1+\phi}) - \frac{1}{1+\phi} \cdot \log_2\phi$ simplifies to $\log_2\phi$ using parts~(b) and (c) of Lemma~\ref{lem:alpha}: replace $\frac{1}{1+\phi}$ and $\log_2\phi$ by $\frac{\alpha}{1-\alpha}$ and $(1-\alpha) \sfh(\frac{\alpha}{1-\alpha})$, respectively, and simplify. Thus, we have
\beq
A(\tau) = \begin{cases}
\sfh(\alpha^{-1}\tau) - \alpha^{-1}\tau\log_2\phi & \text{ if } \alpha^{-1} \tau \le \frac{1}{1+\phi} \\
\log_2\phi & \text{ if } \alpha^{-1} \tau \ge \frac{1}{1+\phi}
\end{cases}
\label{eq:A2}
\eeq

As a result, when $\alpha^{-1} \tau \le \frac{1}{1+\phi}$, we have $A(\tau) = \sfh(\alpha^{-1}\tau) - \alpha^{-1}\tau \log_2\phi \le B(\tau)$. Thus, \eqref{eq:Rtau1} reduces to $R(\tau) \le B(\tau)$, which proves one half of Theorem~\ref{thm:Rupbnd1}.

To complete the proof of the theorem, we must show that when $\alpha^{-1} \tau \ge \frac{1}{1+\phi}$, we have $A(\tau) \ge B(\tau)$. This would then imply that $\max\{A(\tau),B(\tau)\} = A(\tau) = \log_2\phi$ by \eqref{eq:A2}. The above clearly holds when $\alpha^{-1}\tau > 1$, since $B(\tau) = 0$ in this case; so we henceforth assume $1 \ge \alpha^{-1} \tau \ge \frac{1}{1+\phi}$.

We will show that the maximum in the definition of $B(\tau)$ is achieved at $\rho = \alpha^{-1} \tau$. With this, $B(\tau) = \sfh(\alpha^{-1} \tau) - \alpha^{-1} \tau \log_2\phi  \le {\displaystyle \max_{0\le\rho\le \alpha^{-1} \tau}} [\sfh(\rho) - \rho \log_2\phi] = A(\tau)$.

Define $f_{\tau}(\rho) = \sfh(\rho) - (\rho-\tau)\sfh(\frac{\tau}{\rho-\tau})$, so that $B(\tau) = \max_{\alpha^{-1} \tau \le \rho \le 1} f_\tau(\rho)$. We want to show that, under the assumption $1 \ge \alpha^{-1} \tau \ge \frac{1}{1+\phi}$, the function $f_\tau(\rho)$ is monotonically decreasing in the range $\alpha^{-1} \tau \le \rho \le 1$. We accomplish this by showing that $f_\tau'(\alpha^{-1} \tau) \le 0$, and $f''_\tau(\rho) < 0$ for $\alpha^{-1} \tau \le \rho \le 1$. Here, all derivatives are with respect to the variable $\rho$.

\medskip

\underline{$f_\tau'(\alpha^{-1} \tau) \le 0$}: Computing the derivative $f_\tau'(\rho)$ by direct differentiation, then plugging in $\rho = \alpha^{-1}\tau$ and simplifying using Lemma~\ref{lem:alpha}, we obtain
 $$f_\tau'(\alpha^{-1} \tau) = \log_2\frac{1-\alpha^{-1} \tau}{\alpha^{-1} \tau} - \log_2\phi = g'(\alpha^{-1} \tau),$$
where $g$ is the function defined by $g(x) = \sfh(x) - (\log_2\phi) x$. Observe that $g(x)$ is strictly concave on $[0,1]$, and attains its unique maximum at $x = \frac{1}{1+\phi}$. Hence, for $x \ge \frac{1}{1+\phi}$, $g'(x) \le 0$. In particular, $g'(\alpha^{-1} \tau) \le 0$. 

\medskip

\underline{$f_\tau''(\rho) < 0$ for $\alpha^{-1} \tau \le \rho \le 1$}: Routine differentiation yields
$$
f_\tau''(\rho) = - \frac{1}{1-\rho} - \frac{1}{\rho} + \frac{\tau}{(\rho-\tau)(\rho-2\tau)}.
$$
For $\tau \le \alpha \rho$, we have
\begin{align*}
 \frac{\tau}{(\rho-\tau)(\rho-2\tau)} & \ \le \ \frac{\alpha\rho}{(\rho-\alpha\rho)(\rho-2\alpha\rho)} \\
& \ = \ \frac{\alpha}{(1-\alpha)(1-2\alpha)} \cdot \frac{1}{\rho}  \ = \ \frac{5-\sqrt{5}}{1+\sqrt{5}} \cdot \frac{1}{\rho} \ < \ \frac{1}{\rho}.
\end{align*}
Hence, $f_\tau''(\rho) <  - \frac{1}{1-\rho} < 0$.

\medskip

This completes the proof of Theorem~\ref{thm:Rupbnd1}, modulo the promised proof of Lemma~\ref{lem:alpha}(c).

\medskip

\emph{Proof of Lemma~\ref{lem:alpha}(c)\/}:  We first write
$$
(1-\alpha)\sfh\left(\frac{\alpha}{1-\alpha}\right) = 
  - \alpha \log_2 (\alpha(1-\alpha)) - (1-2\alpha) \log_2(1-2\alpha) + \log_2(1-\alpha).
$$
Using $\alpha(1-\alpha) = \frac{1}{5}$ and $1-2\alpha = \frac{1}{\sqrt{5}}$, the right-hand side above simplifies to
$$
\frac12 \log_2 5 + \log_2(1-\alpha) \ = \ \log_2[\sqrt{5}(1-\alpha)].
$$
It is easy to verify that $\sqrt{5}(1-\alpha) = \phi$. 
\qed

\section*{Acknowledgement} The authors thank Artyom Sharov and Ronny Roth for pointing out that the upper bound of Theorem~\ref{thm:Rupbnd2} could be partially improved by the approach of Section~\ref{sec:Rtau_upbnd1}.


\begin{thebibliography}{99}
\bibitem{KK12} A.A.~Kulkarni and N.~Kiyavash, ``Non-asymptotic upper bounds for deletion correcting codes,''
arXiv:1211.3128, Nov.~2012.

\bibitem{Ber79}
C.~Berge, ``Packing Problems and Hypergraph Theory: A Survey,''
\emph{Annals of Discrete Mathematics}, vol.\ 4, pp.\ 3--37, 1979.

\bibitem{MBK11} A.~Mazumdar, A.~Barg and N.~Kashyap, ``Coding for high-density recording on a 1-d granular magnetic medium,'' \emph{IEEE Trans.\ Inform.\ Theory}, vol.\ 57, no.\ 11, pp.\ 7403--7417, Nov.\ 2011.

\bibitem{Pan_etal} L.~Pan, W.E.~Ryan, R.~Wood and B.~Vasic, ``Coding and detection for rectangular grain models,''
\emph{IEEE Trans.\ Magn.}, vol.\ 47, no.\ 6, pp.\ 1705--1711,  June 2011.

\bibitem{SR11} A.~Sharov and R.M.~Roth, ``Bounds and constructions for granular media coding,''
\emph{Proc.\ 2011 IEEE Int.\ Symp.\ Inform.\ Theory (ISIT 2011)}, pp.\ 2304--2308.

\bibitem{SR13} A.~Sharov and R.M.~Roth, ``Bounds and constructions for granular media coding,''
submitted to \emph{IEEE Trans.\ Inform.\ Theory}, 2013.

\bibitem{Wood_etal} R.~Wood, M.~Williams, A.~Kavcic and J.~Miles, ``The feasibility of magnetic recording at 
10 Terabits per square inch on conventional media,'' \emph{IEEE Trans.\ Magn.}, vol.\ 45, no.\ 2, pp.\ 917--923, Feb.\ 2009.


\end{thebibliography}
\end{document}